\documentclass[fleqn,12pt]{article}

\usepackage[latin2]{inputenc}
\usepackage{t1enc}
\usepackage{amssymb}
\usepackage{amsmath}
\usepackage[mathscr]{eucal}
\usepackage{amsthm}
\usepackage{array}
\usepackage{mathbbol,bbm}
\usepackage{color}

\swapnumbers
\theoremstyle{definition} 
\newtheorem{defin}{Definition}[section]
\newtheorem{thm}[defin]{Theorem}

\newtheorem{lemma}[defin]{Lemma}
\newtheorem{prop}[defin]{Proposition}

\def\vfi{\varphi}
\def\hil{{\mathcal H}}

\def\B{{\mathcal B}}
\def\C{{\mathcal C}}
\def\F{{\mathcal F}}

\def\ep{\varepsilon}
\def\N{\mathbb{N}}
\def\iC{\mathbb{C}}
\def\R{\mathbb{R}}
\def\Z{\mathbb Z}
\def\bz{\left(}
\def\jz{\right)}

\def\kii{\emph}
\def\kiii{}
\def\dx{\text{d}x}
\def\egy{\mathbf 1}
\def\unit{\Eins}

\def\dimen{\nu}
\def\speed{\frac{1}{n^{\dimen}}}
\def\dd{\mathrm{d}}

\newcommand{\ki}{\emph}

\newcommand{\s}{\mbox{ }}
\newcommand{\ds}{\mbox{ }\mbox{ }}
\newcommand{\fv}{\hat}
\newcommand{\norm}[1]{\left\| #1\right\|}
\newcommand{\snorm}[1]{\| #1 \|}
\newcommand{\inner}[2]{\langle #1 , #2\rangle}
\newcommand{\abs}[1]{\left| #1 \right|}

\newcommand{\vect}[1]{\mathbf{#1}}

\newcommand{\diad}[2]{|#1\rangle\langle #2|}
\newcommand{\pr}[1]{\diad{#1}{#1}}
\newcommand{\D}{\hat}

\newcommand{\Fock}[1]{\F\bz #1\jz}
\newcommand{\sr}[2]{S\bz #1\,||\, #2\jz}
\newcommand{\srm}[2]{S_{\dimen,M}\bz #1\,||\, #2\jz}
\newcommand{\chernoff}[2]{c_{\dimen}\bz #1\,||\,#2\jz}
\newcommand{\chernoffli}[2]{\underline c_{\dimen}\bz #1\,||\,#2\jz}
\newcommand{\chernoffls}[2]{\overline c_{\dimen}\bz #1\,||\,#2\jz}
\newcommand{\hli}[3]{\underline h_{\dimen}\bz #1|\, #2\,||\,#3\jz}
\newcommand{\hls}[3]{\overline h_{\dimen}\bz #1|\, #2\,||\,#3\jz}
\newcommand{\hlim}[3]{h_{\dimen}\bz #1|\, #2\,||\,#3\jz}
\newcommand{\sli}[2]{\underline s_{\dimen}\bz #1\,||\,#2\jz}
\newcommand{\sls}[2]{\overline s_{\dimen}\bz #1\,||\,#2\jz}
\newcommand{\slim}[2]{s_{\dimen}\bz #1\,||\,#2\jz}
\newcommand{\derleft}[1]{\partial^{-} #1}

\newcommand{\chbound}[2]{C\bz #1\,||\,#2\jz}

\newcommand{\chboundm}[2]{C_{\dimen,M}\bz #1\,||\,#2\jz}
\newcommand{\hbound}[3]{H\bz #1|\, #2\,||\,#3\jz}
\newcommand{\hboundm}[3]{H_{\dimen,M}\bz #1|\, #2\,||\,#3\jz}

\DeclareMathOperator{\Tr}{Tr}
\DeclareMathOperator{\supp}{supp}

\DeclareMathOperator{\spa}{span}

\begin{document}

\centerline{\huge Asymptotic distinguishability measures}
\medskip

 \centerline{\huge for shift-invariant quasi-free states}
\medskip
 \centerline{\huge of fermionic lattice systems}

\bigskip
\s

\bigskip

 \centerline{\large Mil\'an Mosonyi\footnote{Electronic mail: milan.mosonyi@gmail.com}, Fumio Hiai\footnote{Electronic mail: hiai@math.is.tohoku.ac.jp}, Tomohiro Ogawa\footnote{Electronic mail: ogawa@quantum.jst.go.jp}, Mark Fannes\footnote{Electronic mail: mark.fannes@fys.kuleuven.be}}
  \bigskip

\centerline{$^{1,2}$\textit{Graduate School of Information Sciences, Tohoku University}}

\centerline{\textit{Aoba-ku, Sendai 980-8579, Japan}}
\bigskip

\centerline{$^{3}$\textit{PRESTO, Japan Science and Technology Agency}}

\centerline{\textit{4-1-8 Honcho Kawaguchi, Saitama, 332-0012, Japan}}
\bigskip
\centerline{$^{4}$\textit{Institute for Theoretical Physics, K.~U.~Leuven}}

\centerline{\textit{Celestijnenlaan 200D, B-3001 Leuven, Belgium}}
\bigskip

\s

\medskip

\begin{abstract}
We apply the recent results of F.~Hiai, M.~Mosonyi and T.~Ogawa [arXiv:\\0707.2020, to appear in J.~Math.~Phys.] to the asymptotic hypothesis testing problem of locally faithful shift-invariant quasi-free states on a CAR algebra.
We use a multivariate extension of Szeg\H o's theorem to show the existence of the mean Chernoff and Hoeffding bounds and the mean relative entropy, and show that these quantities arise as the optimal error exponents in suitable settings.
\end{abstract}

 \section{Introduction}

Assume that we have a sequence of finite-level quantum systems with Hilbert spaces $\vec\hil:=\{\hil_n\,:\,n\in\N\}$, and that we know a priori that either the $n$th system is in the state $\rho_n$ for each $n\in\N$ (null-hypothesis $H_0$), or in the state $\sigma_n$ (alternative hypothesis $H_1$).
The protocol to decide between these two options is to make a binary positive operator-valued measurement $(T_n,I_n-T_n),\,0\le T_n\le I_n$ on $\hil_n$ for some $n\in\N$. If the outcome corresponding to $T_n$ occurs then $H_0$ is accepted, otherwise it is rejected. Obviously, there are two ways to make an erroneous decision:
to accept $H_0$ when it is false (\ki{error of the first kind}) and to reject it when it is true (\ki{error of the second kind}). The corresponding error probabilities are
\begin{eqnarray*}
\alpha_n (T_n) &:=& \rho_n(I_n-T_n)=\Tr\D{\rho}_n (I_n-T_n)\,,\\
\beta_n (T_n) &:=&  \sigma_n(T_n)=\Tr\D{\sigma}_n T_n\,,
\end{eqnarray*}
where $\D{\rho}_n$ and $\D{\sigma}_n$ denote the densities of $\rho_n$ and $\sigma_n$.
Apart from the trivial situation when $\supp\D{\rho}_n\perp\supp\D{\sigma}_n$, there is no measurement making both error probabilities simultaneously $0$. However, when the measurements are chosen in an optimal way, the error probabilities are expected to vanish with an exponential speed as $n$ goes to infinity. The main goal in the study of asymptotic hypothesis testing is to identify the rate of exponential decay to zero in various settings. Here we will study the following quantities, related to the \ki{Chernoff bound}, the \ki{Hoeffding bound} and \ki{Stein's lemma}, respectively:
\begin{eqnarray*}
 \chernoff{\vec\rho}{\vec\sigma}&:=&\sup_{\{T_n\}}\left\{\lim_{n\to\infty}-\speed\log\big(
\alpha_n(T_n)+\beta_n(T_n)\big)\right\}\,, \\
\hlim{r}{\vec\rho}{\vec\sigma}&:=&\sup_{\{T_n\}}\left\{\lim_{n\to\infty} -\speed\log\beta_n(T_n)\biggm|
\limsup_{n\to\infty} \speed\log\alpha_n(T_n) < -r\right\}\,,\ds r\ge0, \label{hb}\\
\slim{\vec\rho}{\vec\sigma}&:=&\sup_{\{T_n\}}\left\{ \lim_{n\to\infty}-\speed\log\beta_n(T_n)\biggm| \lim_{n\to\infty}\alpha_n(T_n)=0\right\}\,, 
\end{eqnarray*}
where $\nu$ is a fixed positive number and the suprema are taken over sequences of measurements for which the limits exist. In our cases of interest, 
$\nu$ is the physical dimension of an infinite lattice system, and  
$\rho_n$ and $\sigma_n$ are restrictions of translation-invariant states $\rho$ and $\sigma$ of the infinite system to $\nu$-dimensional cubes with sidelength $n$.

It has been shown in various settings \cite{Aud,ANSzV,BS,Hayashi,HMO,HMO2,HP-1,HP,Nagaoka,NSz,ON} that, under suitable conditions, the above exponents coincide with certain relative entropy-like quantities. Apart from giving computable closed expressions for the error exponents, the importance of these results lies in providing an operational interpretation for the given relative entropy-like quantities. 
Our aim in this paper is to establish the equality of the error exponents and the corresponding relative entropy-like quantitites for locally faithful quasi-free states on a CAR algebra, based on the results of \cite{HMO2}.
In Section \ref{section:preliminaries} we give a brief overview on hypothesis testing and quasi-free states and in Section \ref{section:quasifree_testing} we give our main results, relying on \cite{HMO2} and an extension of Szeg\H o's theorem that we prove in Section \ref{section:Szego}.

\section{Preliminaries}\label{section:preliminaries}
\subsection{Hypothesis testing in an asymptotic framework}\label{subsection:hypotesting}

Let $\omega$ and $\vfi$ be states on a finite-dimensional Hilbert space $\hil$ with densities $\D{\omega}$ and $\D{\vfi}$; i.e., $\omega(A)=\Tr\D{\omega}A,\,A\in\B(\hil)$, and similarly for $\vfi$. When $\supp\D{\omega}\le\supp\D{\vfi}$, the
%\ki{R\'enyi relative entropies}, the 
\ki{Chernoff bound}, the \ki{Hoeffding bound(s)} and the \ki{relative entropy} of $\omega$ and $\vfi$ are   defined as
\begin{eqnarray*}
%\chboundt{\omega}{\vfi}&:=&-\log\Tr\D{\omega}^t\D{\vfi}^{1-t}\,,\ds\ds\ds t\in (0,1)\\
\chbound{\omega}{\vfi}&:=&-\min_{0\le t\le 1}\left\{\log\Tr\D{\omega}^t\D{\vfi}^{1-t}\right\}\,,\\
\hbound{r}{\omega}{\vfi}&:=&\sup_{0\le t< 1} \frac{-tr-\log\Tr\D{\omega}^t\D{\vfi}^{1-t}}{1-t}\,,\ds\ds\ds r\ge0,\\
\sr{\omega}{\vfi}&:=&
%\begin{cases}
 \Tr\D{\omega}\bz\log\D{\omega}-\log\D{\vfi}\jz. \\
 %+\infty\,,&\text{otherwise}\,.
%\end{cases}\\
\end{eqnarray*}
(Here we use the conventions $0^t:=0,\,t\in\R$ and $\log 0:=-\infty$.)
All these quantities are non-negative, jointly convex in the variables $\omega,\vfi$, and monotonically decreasing under the simultaneous application of a trace-preserving completely positive map on $\omega$ and $\vfi$ \cite{LR,Petz,Uhlmann}. The 
%R\'enyi relative entropies, the 
Chernoff bound and the relative entropy are also strictly positive, i.e., they vanish only when the two states are equal. The Hoeffding bounds, however, only take strictly positive values on a range of $r$ that depends on the states $\omega$ and $\vfi$; if the supports are equal then this range is easily seen to be $0<r<\sr{\vfi}{\omega}$. Note that $\sup$ can be replaced with $\max$ in the definition of the Hoeffding bound for $r>0$, and  $\hbound{0}{\omega}{\vfi}=\sr{\omega}{\vfi}$.

Now let $\vec\rho:=\{\rho_n\}_{n\in\N}$ and $\vec\sigma:=\{\sigma_n\}_{n\in\N}$ be two sequences of states on the finite-dimensional Hilbert spaces $\vec\hil:=\{\hil_n\}_{n\in\N}$. We assume throughout the paper that $\supp\D{\rho}_n\le\supp\D{\sigma}_n,\,n\in\N$. With the above conventions, the functions
\begin{equation*}
\psi_n(t):=\log\Tr\D{\rho}_n^t\D{\sigma}_n^{1-t}
\end{equation*}
are well-defined for all $t\in\R$, and they are easily seen to be convex on $\R$, with the properties
\begin{equation*}
\psi_n(1)=0 \ds\ds\ds \text{and}\ds\ds\ds \psi_n'(1)=\sr{\rho_n}{\sigma_n}\,.
\end{equation*}
If $\supp\D{\rho}_n=\supp\D{\sigma}_n$ then also $\psi_n(0)=0$ and $\psi_n'(0)=-\sr{\sigma_n}{\rho_n}$.
%$\psi_n(0)=\psi_n(1)=0$, and $\psi_n'(0)=\sr{\sigma_n}{\rho_n}$ and $\psi_n'(1)=\sr{\rho_n}{\sigma_n}$.
We define the mean versions of the Chernoff bound, the Hoeffding bounds and the relative entropy by
\begin{eqnarray}
\chboundm{\vec\rho}{\vec\sigma}&:=&\lim_{n\to\infty}\speed\chbound{\rho_n}{\sigma_n}\,,\\
\hboundm{r}{\vec\rho}{\vec\sigma}&:=&\lim_{n\to\infty}\speed\hbound{n^{\dimen}r}{\rho_n}{\sigma_n}\,,\\
\srm{\vec\rho}{\vec\sigma}&:=&\lim_{n\to\infty}\speed\sr{\rho_n}{\sigma_n}\,,\label{def:srm}
\end{eqnarray}
for some positive number $\nu$,
whenever the limits exist. One can easily see that if the functions $\speed\psi_n$ converge uniformly to some function $\psi$ on $[0,1]$ then the mean Chernoff bound and the mean Hoeffding bounds exist, and
\begin{eqnarray}
\chboundm{\vec\rho}{\vec\sigma}&=&-\min_{0\le t\le 1}\psi(t)\,,\label{chboundm}\\
\hboundm{r}{\vec\rho}{\vec\sigma}&=&\max_{0\le t< 1} \frac{-tr-\psi(t)}{1-t}\,,\ds\ds\ds r>0\,.\label{hboundm}
 \end{eqnarray}
Moreover, if the left derivatives $\speed\derleft{\psi_n}(1)$ converge to $\derleft{\psi}(1)$ then
\begin{equation}\label{mrelentr}
\srm{\vec\rho}{\vec\sigma}=\hboundm{0}{\vec\rho}{\vec\sigma}=\derleft{\psi}(1)=\sup_{0\le t< 1} \frac{-\psi(t)}{1-t}\,.
\end{equation}
%if $\supp\D{\rho}_n\le \supp\D{\sigma}_n,\, n\in\N$, then
% $\srm{\vec\rho}{\vec\sigma}=\lim_{n\to\infty}\speed\psi_n'(1)$ whenever the latter limit exists. 
Note that the convexity of the functions $\speed\psi_n$ implies that 
if they converge pointwise to a function $\psi$ on some open set $G$ then $\psi$ is also convex, and, moreover, the
 convergence is uniform on any compact subinterval of $G$.
 
Let $\chernoff{\vec\rho}{\vec\sigma},\hlim{r}{\vec\rho}{\vec\sigma}$ and $\slim{\vec\rho}{\vec\sigma}$ be the error exponents given in the Introduction.
We also define the underlined and overlined versions of these quantities,
 by replacing the limits with liminf and limsup, respectively; i.e.,
%~$\sli:=\inf_{\{T_n\}}\big\{ \liminf_{n\to\infty}\speed\log\beta_n(T_n)\,|\, \lim_{n\to\infty}\alpha_n(T_n)=0\big\}$ and $\sls:=\inf_{\{T_n\}}\big\{ \limsup_{n\to\infty}\speed\log\beta_n(T_n)\,|\, \lim_{n\to\infty}\alpha_n(T_n)=0\big\}$. 
\begin{eqnarray*}
 \chernoffli{\vec\rho}{\vec\sigma}&:=&\sup_{\{T_n\}}\left\{\liminf_{n\to\infty}-\speed\log\big(
\alpha_n(T_n)+\beta_n(T_n)\big)\right\}\,, \\
 \chernoffls{\vec\rho}{\vec\sigma}&:=&\sup_{\{T_n\}}\left\{\limsup_{n\to\infty}-\speed\log\big(
\alpha_n(T_n)+\beta_n(T_n)\big)\right\}\,, \\
\hli{r}{\vec\rho}{\vec\sigma}&:=&\sup_{\{T_n\}}\left\{\liminf_{n\to\infty}- \speed\log\beta_n(T_n)\biggm|
\limsup_{n\to\infty} \speed\log\alpha_n(T_n) < -r\right\}\,,\ds r\ge0, \\
\hls{r}{\vec\rho}{\vec\sigma}&:=&\sup_{\{T_n\}}\left\{\limsup_{n\to\infty}- \speed\log\beta_n(T_n)\biggm|
\limsup_{n\to\infty} \speed\log\alpha_n(T_n) < -r\right\}\,,\ds r\ge0, \\
\sli{\vec\rho}{\vec\sigma}&:=&\sup_{\{T_n\}}\left\{\liminf_{n\to\infty}-\speed\log\beta_n(T_n)\biggm| \lim_{n\to\infty}\alpha_n(T_n)=0\right\}\,, \\
\sls{\vec\rho}{\vec\sigma}&:=&\sup_{\{T_n\}}\left\{\limsup_{n\to\infty}-\speed\log\beta_n(T_n)\biggm| \lim_{n\to\infty}\alpha_n(T_n)=0\right\}\,.
\end{eqnarray*}
Obviously,
$\chernoff{\vec\rho}{\vec\sigma}\le\chernoffli{\vec\rho}{\vec\sigma}\le\chernoffls{\vec\rho}{\vec\sigma},\s\hlim{r}{\vec\rho}{\vec\sigma}\le\hli{r}{\vec\rho}{\vec\sigma}\le\hls{r}{\vec\rho}{\vec\sigma},\\ r\ge0,$ and $\slim{\vec\rho}{\vec\sigma}\le\sli{\vec\rho}{\vec\sigma}\le\sls{\vec\rho}{\vec\sigma}$ for any $\nu>0$.
% Note that $\chernoff$ can be expressed as $\lim_{n\to\infty}\speed\log\min_{\{T_n\}}\{\alpha_n(T_n)+\beta_n(T_n)\}$ whenever the latter quantity exists. 

The following theorem, given in \cite{HMO2}, connects the mean Chernoff and Hoeffding bounds and the mean relative entropy to the corresponding error exponents. Since the aim in \cite{HMO2} was to study the hypothesis testing problem for one-dimensional spin chains, only the case $\dimen=1$ was considered. All the proofs, however, go through unaltered for any $\dimen>0$.
\begin{thm}\label{thm:error exponents}
Assume that the limit
\begin{equation}\label{def:psi}
\psi(t):=\lim_{n\to\infty}\speed\psi_n(t)
\end{equation}
exists as a real number for all $t\in\R$. Assume, moreover, that $\psi$ is differentiable on $\R$ and $\lim_{n\to\infty}\speed\psi_n'(1)=\psi'(1)$. Then the mean Chernoff and Hoeffding bounds and the mean relative entropy exist, the relations \eqref{chboundm}, \eqref{hboundm} and \eqref{mrelentr} hold, and  
\begin{eqnarray*}
\chernoff{\vec\rho}{\vec\sigma}=\chernoffli{\vec\rho}{\vec\sigma}=\chernoffls{\vec\rho}{\vec\sigma}&=&\chboundm{\vec\rho}{\vec\sigma},\\
\hlim{r}{\vec\rho}{\vec\sigma}=\hli{r}{\vec\rho}{\vec\sigma}=\hls{r}{\vec\rho}{\vec\sigma}&=& \hboundm{r}{\vec\rho}{\vec\sigma}\,,\ds\ds\ds r\ge0,\\
\slim{\vec\rho}{\vec\sigma}=\sls{\vec\rho}{\vec\sigma}=\sls{\vec\rho}{\vec\sigma}&=&\srm{\vec\rho}{\vec\sigma}\,.
\end{eqnarray*} 
\end{thm}

 \subsection{Quasi-free states}\label{subsection:quasifree}

 Our general reference for this section is \cite{F}.
  Let $\hil$ be a separable Hilbert space and $\F(\hil):=\oplus_{k=0}^{\dim\hil}\wedge^k \hil$ be the corresponding antisymmetric Fock space, with the convention $\wedge^0\hil:=\iC$. We use the notation 
\begin{equation*}
x_1\wedge\ldots\wedge x_k:=\frac{1}{\sqrt{n!}}\sum_{\sigma\in S_k} s(\sigma) x_{\sigma(1)}\otimes\ldots\otimes x_{\sigma(k)}\,,\ds\ds\ds x_1,\ldots,x_k\in\hil,
\end{equation*}
where the summation runs over all permutations of $k$ elements and $s(\sigma)$ denotes the sign of the permutation $\sigma$. For each $x\in\hil$ the corresponding \ki{creation operator} is defined as the unique bounded linear extension $ c^*(x):\,\F(\hil)\to\F(\hil)$ of 
\begin{equation*}
c^*(x):\,x_1\wedge\ldots\wedge x_k\mapsto x\wedge x_1\wedge\ldots\wedge x_k\,,\ds\ds\ds x_1,\ldots,x_k\in\hil,\s k\in\N,
\end{equation*}
and the corresponding \ki{annihilation operator} is its adjoint $c(x):=\bz c^*(x)\jz^*$.
Creation and annihilation operators satisfy the \ki{canonical anticommutation relations (CAR)}:
\begin{equation*}
 c(x)c(y)+c(y)c(x)=0\,,\ds\ds\ds c(x)c^*(y)+c^*(y)c(x)=\inner{x}{y}\unit\,.
\end{equation*}
The $C^*$-algebra generated by $\{c(x)\,:\,x\in\hil\}$ is called the algebra of the canonical commutation relations, and is denoted by CAR($\hil$).
The von Neumann algebra generated by $\{c(x)\,:\,x\in\hil\}$ is equal to $\B(\F(\hil))$; in particular, for a finite-dimensional Hilbert space $\hil$ we have CAR($\hil$)=$\B(\F(\hil))$.
 
If $A\in\B(\hil)$ then $A^{\otimes k}$ leaves the antisymmetric subspace of $\hil^{\otimes k}$ invariant. We denote the restriction of $A^{\otimes k}$ onto $\wedge^k \hil$ by $\wedge^k A$, and introduce the notation 
\begin{equation*}
\Fock{A}:=\oplus_{k=0}^{\dim\hil} \wedge^k A\,.
\end{equation*}
Here we use the convention $A^{\otimes 0}:=\wedge^{0}A:=1$.
This yields a bounded operator if $\hil$ is finite-dimensional or if $\norm{A}\le 1$.
If $\hil$ is finite-dimensional and $A$ has eigenvalues $\lambda_1,\ldots,\lambda_d$, counted with multiplicities, then the eigenvalues of $\wedge^k A$ are 
$\{\lambda_{i_1}\cdot\ldots\cdot\lambda_{i_k}\,:\,i_1<\ldots<i_k\}$. Thus we get that in this case
\begin{equation*}
\Tr\Fock{A}=\det(I+A)\,.
\end{equation*}

Let $Q\in\B(\hil)$ and define a functional $\omega_{Q}$ on monomials by
\begin{equation}
    \omega_{Q}\,\bz c(x_1)^*\ldots c(x_n)^* c(y_m)\ldots c(y_1)\jz= \delta_{m,n}\det \{\inner{y_i}{Q\s x_j}\}_{i,j=1}^n\,.
\end{equation}
If $0\le Q\le I$, then $\omega_Q$ extends to a state on CAR($\hil$). Such states are called 
 \ki{quasi-free}. For a state $\omega_{Q}$ the operator $Q$ is called the \ki{symbol} of $\omega_Q$. When $\hil$ is finite-dimensional, we have the following \cite[Lemma 3]{DFP}:
\begin{prop}\label{density}
Let $\hil$ be finite-dimensional and $Q\in\B(\hil)$ be a symbol, and assume that $1$ is not an eigenvalue of $Q$. Then the density of the corresponding quasi-free state is
\begin{equation*}
\D{\omega}_{Q}=\det(I-Q)\oplus_{k=0}^{\dim\hil}\wedge^k\frac{Q}{I-Q}=\det(I-Q)\F\bz\frac{Q}{I-Q}\jz\,.
\end{equation*}
\end{prop}
Quasi-free states emerge as equilibrium states of non-interacting fermionic systems. For instance, if the  one-particle Hamiltonian $H$ of a system of non-interacting fermions is such that $e^{-\beta H}$ is trace-class then 
the Gibbs state of the system at inverse temperature $\beta$ is the quasi-free state with symbol
$Q=\frac{e^{-\beta H}}{I+e^{-\beta H}}$ (see, e.g., \cite[Proposition 5.2.23]{BR2}).
 
Consider now a $\dimen$-dimensional fermionic lattice system with Hilbert space $l^2(\Z^{\dimen})$. We denote the standard basis of $l^2(\Z^{\dimen})$ by $\{\egy_{\vect{k}}\,:\,\vect{k}\in\Z^{\dimen}\}$, and define the \ki{shift operators} as the unique linear extensions of $S_{\vect{j}}\egy_{\vect{k}}\mapsto \egy_{\vect{k}+\vect{j}},\s \vect{k}\in\Z^{\dimen}$, for all $\vect{j}\in\Z^{\dimen}$.
The map $\gamma_{\vect{j}}(c(x)):=c\bz S_\vect{j}x\jz$ extends to an automorphism of CAR$\bz l^2(\Z^{\dimen})\jz$ for all $\vect{j}\in\Z^{\dimen}$,
and $\gamma_{\vect{j}},\s \vect{j}\in\Z^{\dimen}$, is a group of automorphisms, called the group of \ki{shift 
automorphisms}. A quasi-free state $\omega_Q$ is called shift-invariant if $\omega_{Q}\circ\gamma_\vect{j}=\omega_Q,\s \vect{j}\in\Z^{\dimen}$, which 
holds if and only if its symbol $Q$ is shift-invariant, i.e., it commutes with all the unitaries $S_{\vect{j}},\s  \vect{j}\in\Z^{\dimen}$.

\section{Szeg\H o's theorem}\label{section:Szego}

 Shift-invariant operators on $l^2(\Z^{\dimen})$ commute with each other and they are simultaneously diagonalized by the Fourier transformation
 \begin{equation*}
   F:\,l^2(\Z^{\dimen})\to L^2([0,2\pi)^{\dimen})\,,\ds F\egy_{\{\vect{k}\}}:=\vfi_{\vect{k}}\,,\ds \vfi_{\vect{k}}(\vect{x}):=e^{i\inner{\vect{k}}{\vect{x}}}\,,\ds\vect{x}\in[0,2\pi)^{\dimen}\,,\vect{k}\in\Z^{\dimen}\,,
   \end{equation*}
where $\inner{\vect{k}}{\vect{x}}:=\sum_{i=1}^{\dimen}k_ix_i$. That is,
every shift-invariant operator $A$ arises in the form $A=F^{-1}M_{\fv a}F$, where $M_{\fv a}$ denotes the multiplication operator by a bounded measurable function $\fv a$ on $[0,2\pi)^{\dimen}$. As a consequence, the matrix entries of shift-invariant operators are constants along diagonals; more explicitly, $\inner{\egy_{\{\vect{k}\}}}{A\egy_{\{\vect{j}\}}}=\frac{1}{(2\pi)^{\dimen}}\int_{[0,2\pi)^{\dimen}}e^{-i\inner{\vect{k}-\vect{j}}{\vect{x}}}\fv a(\vect{x})\,\dd \vect{x}$.
Szeg\H o's classic theorem \cite{GS} states that if $\fv a$ is real-valued on $[0,2\pi)$ then 
\begin{equation*}
\lim_n\frac{1}{n}\Tr f(A_n)=\frac{1}{2\pi}\int_0^{2\pi} f(\fv a(x))\,\dx
\end{equation*}
for any continuous function $f$ on the convex hull $\Sigma(A)$ of the spectrum of $A$, where $A_n:=P_nAP_n$ with 
$P_n:=\sum_{k=0}^{n-1}\pr{\egy_{\{k\}}}$.
In higher dimensions, let $P_n:=\sum_{k_1,\ldots,k_{\dimen}=0}^{n-1}\pr{\egy_{\{\vect{k}\}}}$ and $A_n:=P_nAP_n$ for $A=F^{-1}M_{\fv a}F$.
The following is a multivariate generalization of Szeg\H o's theorem, which is also a generalization for higher dimensions:
\begin{lemma}\label{lemma:Szego}
Let $\fv a^{(1)},\ldots,\fv a^{(r)}$ be bounded measurable functions on $[0,2\pi)^{\dimen}$ with corresponding shift-invariant operators $A^{(1)},\ldots,A^{(r)}$. Then,
\begin{equation}\label{convergence}
\lim_{n\to \infty}\frac{1}{n^{\dimen}}\Tr f^{(1)}\bz A_n^{(1)}\jz\cdot\ldots\cdot f^{(r)}\bz A_n^{(r)}\jz=
\frac{1}{(2\pi)^{\dimen}}\int_{[0,2\pi)^{\dimen}} f^{(1)}\bz \fv a^{(1)}(\vect{x})\jz\cdot\ldots\cdot f^{(r)}\bz \fv a^{(r)}(\vect{x})\jz \dd \vect{x}
\end{equation}
for any choice of polynomials $f^{(1)},\ldots,f^{(r)}$. If all $\fv a^{(k)}$ are real-valued then \eqref{convergence} holds
when $f^{(k)}$ is a continuous function on $\Sigma(A^{(k)})$ for all $1\le k\le r$. In this case, the convergence is uniform on norm-bounded subsets of $\prod_{k=1}^n C\bz \Sigma(A^{(k)})\jz$.
\end{lemma}
\begin{proof}
Obviously, the statement for polynomials follows if we can prove that 
\begin{equation}\label{convergence2}
\lim_{n\to \infty}\frac{1}{n^{\dimen}}\Tr A_n^{(1)}\cdot\ldots \cdot A_n^{(r)} =
\frac{1}{(2\pi)^{\dimen}}\int_{[0,2\pi)^{\dimen}}  \fv a^{(1)}(x)\cdot\ldots\cdot  \fv a^{(r)}(x)\,\dx\,,
\end{equation}
for arbitrary $r\in\N$ and bounded measurable functions $\fv a^{(1)},\ldots,\fv a^{(r)}$.
First, let $\fv a^{(k)}(x)=e^{i\inner{\vect{p}_k}{\vect{x}}},\,1\le k\le r,$ with $\vect{p}_1,\ldots,\vect{p}_r\in\Z^{\dimen}$.
Then $\inner{\egy_{\{\vect{i}\}}}{A^{(k)}\egy_{\{\vect{j}\}}}=\delta_{\vect{i}-\vect{j},\vect{p}_k}$ and therefore the diagonal elements of 
$A_n^{(1)}\cdot\ldots\cdot A_n^{(r)}$ are all $0$, unless $\vect{p}_1+\ldots+\vect{p}_r=0$, in which case the diagonal consists of $0$'s and $1$'s, and the number of the $1$'s is between $n^{\dimen}-|\vect{p}_1|-\ldots-|\vect{p}_r|$ and $n^{\dimen}$. Thus
\begin{equation*}
\lim_{n\to \infty}\frac{1}{n^{\dimen}}\Tr  A_n^{(1)}\cdot\ldots\cdot A_n^{(r)} =
\delta_{\vect{p}_1+\ldots+\vect{p}_r,0}
=
\frac{1}{(2\pi)^{\dimen}}\int_{[0,2\pi)^{\dimen}} \fv a^{(1)}(\vect{x})\cdot\ldots\cdot \fv a^{(r)}(\vect{x})\, \dd \vect{x}\,.
\end{equation*}
From this, \eqref{convergence2} follows immediately in the case when $\fv a^{(k)},\,1\le k\le r,$ are trigonometric polynomials.

Now let $\fv a^{(k)},\,1\le k\le r,$  be bounded measurable functions on $[0,2\pi)^{\dimen}$.
One can see (by taking the Fej\'er means of the Fourier series) that for any $\ep>0$ there exist trigonometric polynomials $\fv a^{(k)}_{\ep},\,1\le k\le r,$  such that 
$\snorm{\fv a^{(k)}-\fv a^{(k)}_{\ep}}_2\le\ep$ and $\snorm{\fv a^{(k)}_{\ep}}_{\infty}\le \snorm{\fv a^{(k)}}_{\infty}$  for all $1\le k\le r$. 
Note that if $A=F^{-1}M_{\fv a}F$ then
  \begin{eqnarray}
  \norm{A_n}_2^2 &=& \Tr P_n A^*P_n A P_n \le \Tr P_n A^* A P_n
 =\sum_{k_1,\ldots,k_{\dimen}=0}^{n-1} \inner{\egy_{\{\vect{k}\}}}{A^*A\egy_{\{\vect{k}\}}}\nonumber\\
&=&\sum_{k_1,\ldots,k_{\dimen}=0}^{n-1} \inner{\vfi_{\vect{k}}}{|\fv a|^2\vfi_{\vect{k}}}
 =n^{\dimen}\norm{\fv a}_2^2\,,\label{norm1}
  \end{eqnarray}
and obviously
\begin{equation}\label{norm2}
\norm{A_n}\le\norm{A}=\norm{\fv a}_{\infty}\,.
\end{equation}
As a consequence, we get
for any bounded operators $X,Y$ on $\hil$ and $1\le k\le r$,
\begin{eqnarray*}
\abs{\Tr X A_n^{(k)} Y-\Tr X \bz A_{\ep}^{(k)}\jz_n Y}&\le& 
\norm{X}\norm{Y}\norm{A_n^{(k)}-\bz A_{\ep}^{(k)}\jz_n}_1\\ 
&\le& \norm{X}\norm{Y}\norm{I_n}_2 \norm{A_n^{(k)}-\bz A_{\ep}^{(k)}\jz_n}_2\\
&\le& n^{\dimen}\norm{\fv a^{(k)}-\fv a_{\ep}^{(k)}}_2\norm{X}\norm{Y}\le n^{\dimen}\ep\norm{X}\norm{Y},
\end{eqnarray*}
due to H\"older's inequality and inequality \eqref{norm1}. Iterated application of the above inequality, combined with \eqref{norm2}, yields
\begin{equation*}
\abs{\frac{1}{n^{\dimen}}\Tr  A_n^{(1)}\cdot\ldots\cdot A_n^{(r)} -\frac{1}{n^{\dimen}}\Tr \bz A_{\ep}^{(1)}\jz_n\cdot\ldots\cdot
 \bz A_{\ep}^{(r)}\jz_n }
\le \ep\cdot c\,,
 \end{equation*}
 where
$c:=r\max_{1\le k\le r}\snorm{\fv a^{(k)}}^{r-1}_{\infty}$
% $c=\snorm{\fv a^{(1)}\cdot\ldots\cdot \fv a^{(r)}}_{\infty}\bz 1/\snorm{\fv a^{(1)}}_{\infty}+\ldots+1/\snorm{\fv a^{(1)}}_{\infty}\jz$ 
is independent of $n$. Similarly,
 \begin{equation*}
\abs{\frac{1}{(2\pi)^{\dimen}}\int_{[0,2\pi)^{\dimen}} \fv a^{(1)}(\vect{x})\cdot\ldots\cdot  \fv a^{(r)}(\vect{x})\, \dx-
\frac{1}{(2\pi)^{\dimen}}\int_{[0,2\pi)^{\dimen}}  \fv a^{(1)}_{\ep}(\vect{x})\cdot\ldots\cdot  \fv a^{(r)}_{\ep}(\vect{x})\, \dd \vect{x}
}\le \ep\cdot c\,,
\end{equation*}
and from these and the previous argument, the first assertion follows.

Now if $\fv a^{(k)}$ is real-valued then $A^{(k)}_n$ is self-adjoint for all $n$ and its spectrum is easily seen to be in the convex hull of the spectrum of $A^{(k)}$, hence $f^{(k)}( A^{(k)}_n)$ is well-defined for all $n$. By the Stone-Weierstrass theorem, for each $\ep>0$ there exist polynomials $f^{(1)}_{\ep},\ldots,f^{(r)}_{\ep}$ such that 
$\snorm{f^{(k)}-f^{(k)}_{\ep}}_{\infty}<\ep$ and $\snorm{f^{(k)}}_{\infty}\le\snorm{f^{(k)}_{\ep}}_{\infty}$ for all $1\le k\le r$. For any bounded operators $X,Y$ on $\hil$ and $f,g\in C\bz\Sigma(A^{(k)})\jz$ we have
\begin{eqnarray}
\abs{\Tr X f (A_n^{(k)}) Y-\Tr X g( A^{(k)}_n) Y}&\le& 
\norm{X}\norm{Y}\norm{ f (A_n^{(k)}) -  g( A^{(k)}_n)}_1\nonumber\\ 
&\le& \norm{X}\norm{Y} n^{\dimen}\norm{f (A_n^{(k)}) -  g( A^{(k)}_n)}\nonumber\\
&\le& \norm{X}\norm{Y} n^{\dimen}\norm{ f-g}_{\infty}\,,\label{bound1}
%&=&n\ep\norm{X}\norm{Y}\,,\label{bound1}
\end{eqnarray}
and thus
\begin{equation*}
\abs{\frac{1}{n^{\dimen}}\Tr f^{(1)}( A^{(1)}_n)\cdot\ldots\cdot f^{(r)}( A^{(r)}_n)
 -\frac{1}{n^{\dimen}}\Tr f^{(1)}_{\ep}( A^{(1)}_n)\cdot\ldots\cdot f^{(r)}_{\ep}( A^{(r)}_n)}
\le \ep\cdot c\,, 
 \end{equation*}
where $c:=r\max_{1\le k\le r}\snorm{f^{(k)}}^{r-1}$
%$c=(\snorm{f^{(1)}}_{\infty}+1)\cdot\ldots\cdot(\snorm{f^{(r)}}_{\infty}+1)$ 
is independent of $n$.
Obviously, 
\begin{align*}
&\lim_{\ep\to 0}\frac{1}{(2\pi)^{\dimen}}\int_{[0,2\pi)^{\dimen}}f^{(1)}_{\ep}\bz \fv a^{(1)}( \vect{x})\jz\cdot\ldots\cdot f^{(r)}_{\ep}\bz \fv a^{(r)}( \vect{x})\jz\,\dd \vect{ \vect{x}}\\
&\ds\ds\ds 
=\frac{1}{(2\pi)^{\dimen}}\int_{[0,2\pi)^{\dimen}}f^{(1)}\bz \fv a^{(1)}( \vect{x})\jz\cdot\ldots\cdot f^{(r)}\bz \fv a^{(r)}( \vect{x})\jz\,\dd\vect{x}\,,
\end{align*}
% \begin{equation*}
% \frac{1}{(2\pi)^{\dimen}}\int_{[0,2\pi)^{\dimen}}\abs{f^{(1)}_{\ep}\bz \fv a^{(1)}(x)\jz\cdot\ldots\cdot f^{(r)}_{\ep}\bz \fv a^{(r)}(x)\jz-
% f^{(1)}\bz \fv a^{(1)}(x)\jz\cdot\ldots\cdot f^{(r)}\bz \fv a^{(r)}(x)\jz}\dd\vect{x}\le\ep\cdot c\,,
% \end{equation*}
and from these and the validity of \eqref{convergence} for polynomials, the assertion follows. The uniformity of the convergence in the last statement is an immediate consequence of \eqref{bound1}.
\end{proof}

\section{Hypothesis testing for quasi-free states}\label{section:quasifree_testing}
 
Consider a $\dimen$-dimensional lattice and shift-invariant quasi-free states $\omega_Q$ and $\omega_R$
on CAR$(l^2(\Z^{\dimen}))$, with symbols $Q=F^{-1}M_{\fv q}F$ and $R=F^{-1}M_{\fv r}F$,
where $\fv q,\fv r:\,[0,2\pi)^{\dimen}\to [0,1]$.
% Let $\dimen\in\N$, $\fv q,\fv r:\,[0,2\pi)^{\dimen}\to [0,1]$ and $\omega_Q$ and $\omega_R$ be the corresponding shift-invariant quasi-free states on CAR$(l^2(\Z^{\dimen}))$ with symbols $Q=F^{-1}M_{\fv q}F$ and $R=F^{-1}M_{\fv r}F$. 
Let $\hil_n:=\spa\{\egy_{\vect{k}}\,:\,k_1,\ldots,k_{\dimen}=0,\ldots,n-1\}$, and let
 $\omega_{Q_n}$ and $\omega_{R_n}$ be the restrictions of $\omega_Q$ and $\omega_R$ onto CAR$\bz\hil_n\jz$. We will study the asymptotic hypothesis testing problem for $\{\rho_n\}_{n\in\N}$ vs. $\{\sigma_n\}_{n\in\N}$ with $\rho_n:=\omega_{Q_n}$ and $\sigma_n:=\omega_{R_n}$. That is, the null-hypothesis in this case is that the true state of the infinite system is $\omega_{Q}$, while the alternative hypothesis is that it is $\omega_R$, and we make measurements on local subsystems to decide between these two options. We will replace $\vec \rho$ and $\vec\sigma$ with $\omega_{Q}$ and $\omega_R$ in all notations, as the latter uniquely determine the former.

All over this section we will assume that there exists an $\eta\in(0,1/2)$ such that $\eta\le\fv q,\fv r\le 1-\eta$. As a consequence, the local restrictions are faithful. The core of our main result is the following:
\begin{prop}\label{mre}
The limits in \eqref{def:srm} and \eqref{def:psi} exist, 
\begin{equation}
\srm{\omega_Q}{\omega_R} 
=\frac{1}{(2\pi)^{\dimen}}\,\int_{[0,2\pi)^{\dimen}}\,\left[\fv q(\vect{x})\log\frac{\fv q(\vect{x})}{\fv r(\vect{x})}
+(1-\fv q(\vect{x}))\log\frac{1-\fv q(\vect{x})}{1-\fv r(\vect{x})}\right]\,\dd \vect{x}\,,\ds\ds\ds\label{formula:mre}
\end{equation}
and
\begin{equation}\label{formula:psi}
\psi(t)=\frac{1}{(2\pi)^{\dimen}}\,\int_{[0,2\pi)^{\dimen}}\,\log\left[\fv q(\vect{x})^t\fv r(\vect{x})^{1-t}+(1-\fv q(\vect{x}))^t(1-\fv r(\vect{x}))^{1-t}\right]\,\dd \vect{x}\,,\ds t\in\R.\ds
\end{equation}
Moreover, $\psi$ is differentiable on $\R$, and $\srm{\omega_Q}{\omega_R}=\psi'(1)$.
\end{prop}
\begin{proof}
By Proposition \ref{density}, we have
\begin{eqnarray}\label{formula:relentr}
\sr{\omega_{Q_n}}{\omega_{R_n}}&=&\Tr\big[ Q_n\bz\log Q_n-\log R_n\jz\nonumber\\
& &+(I_n-Q_n)\bz\log(I_n-Q_n)-\log(I_n-R_n)\jz \big]
\end{eqnarray}
and
\begin{equation}
\Tr\D{\omega}_{Q_n}^t\D{\omega}_{R_n}^{1-t}= \det(I_n-Q_n)^t\det(I_n-R_n)^{1-t}\det\left[ I_n+\bz\frac{Q_n}{I_n-Q_n}\jz^t\bz\frac{R_n}{I_n-R_n}\jz^{1-t}\right]\label{formula:psi0}
%&=&\det \left[ Q_n^tR_n^{1-t}+(I_n-Q_n)^t(I_n-R_n)^t\right]
\end{equation}
for each $n\in\N$.
Since $\log$ is continuous on $[\eta,1-\eta]$, \eqref{formula:mre} follows immediately from \eqref{formula:relentr} and Lemma \ref{lemma:Szego}.
To prove \eqref{formula:psi}, note that \eqref{formula:psi0} can be rewritten as
 \begin{equation*}
\frac{1}{n^{\dimen}}\log\Tr\D{\omega}_{Q_n}^t\D{\omega}_{R_n}^{1-t}
= \frac{1}{n^{\dimen}}\Tr\log (I_n-Q_n)^t+\frac{1}{n^{\dimen}}\Tr\log (I_n-R_n)^{1-t}
+\frac{1}{n^{\dimen}}\Tr\log\bz I_n+W_{n,t} \jz\,,
\end{equation*}
where $W_{n,t}:=\bz\frac{Q_n}{I_n-Q_n}\jz^{t/2}\bz\frac{R_n}{I_n-R_n}\jz^{1-t}\bz\frac{Q_n}{I_n-Q_n}\jz^{t/2}$.
By Lemma \ref{lemma:Szego},
\begin{align*}
&\lim_{n\to\infty}\frac{1}{n^{\dimen}}\Tr\left[\log(I_n-Q_n)^t+\log(I_n-R_n)^{1-t}\right]\\
&\ds\ds\ds\ds\ds\ds\ds\ds\ds\ds\ds\ds\ds\ds\ds\ds\ds\ds\ds
=\frac{1}{(2\pi)^{\dimen}}\int_{[0,2\pi)^{\dimen}}\left[\log(1-\fv q(\vect{x}))^t(1-\fv r(\vect{x}))^{1-t}\right]\,\dd\vect{x}\,.
\end{align*}
Now for a fixed $t$ we can choose 
an $M>0$ such that $0\le W_{n,t}\le MI_n$ for every $n$. The Taylor series expansion
$\log(1+x)=\sum_{m=0}^{\infty}c_m(x-M/2)^m$ is absolutely and uniformly convergent on $[0,M]$, and therefore
\begin{equation*}
\abs{\frac{1}{n^{\dimen}}\Tr\log\bz I_n+W_{n,t} \jz-\frac{1}{n^{\dimen}}\Tr\sum_{m=0}^N c_m (W_{n,t}-(M/2)I_n)^m}
\end{equation*}
converges to $0$ uniformly in $n$ as $N\to\infty$. Similarly to the proof of Lemma \ref{lemma:Szego}, it is then enough to show that $\frac{1}{n^{\dimen}}\Tr (W_{n,t}-(M/2)I_n)^m$ converges to $\frac{1}{(2\pi)^{\dimen}}\int_{[0,2\pi)^{\dimen}} (h_t(\vect{x})-(M/2))^m\,\dd \vect{x}$ for each $m\in\N$, where $h_t(\vect{x}):=\bz \fv q(\vect{x})/(1-\fv q(\vect{x}))\jz^t\bz \fv r(\vect{x})/(1-\fv r(\vect{x}))\jz^{1-t}$. This, however, follows immediately from Lemma \ref{lemma:Szego}. The differentiability of $\psi$ and $\srm{\omega_Q}{\omega_R}=\psi'(1)$ follow from \eqref{formula:psi} and \eqref{formula:mre} by a straightforward computation.
\end{proof}

Now we can state the main result of our paper:
\begin{thm}\label{main result}
Let $\omega_Q$ and $\omega_R$ be quasi-free states as above. 
The mean Chernoff and Hoeffding bounds and the mean relative entropy exist, the relations \eqref{chboundm}, \eqref{hboundm} and \eqref{mrelentr} hold with $\psi$ given in \eqref{formula:psi}, and
\begin{eqnarray*}
\chernoff{\omega_Q}{\omega_R}=\chernoffli{\omega_Q}{\omega_R}=\chernoffls{\omega_Q}{\omega_R}&=&\chboundm{\omega_Q}{\omega_R}\,,\\
\hlim{r}{\omega_Q}{\omega_R}=\hli{r}{\omega_Q}{\omega_R}=\hls{r}{\omega_Q}{\omega_R}&=&\hboundm{r}{\omega_Q}{\omega_R} \,,\ds\ds r\ge 0\,,\\
\slim{\omega_Q}{\omega_R}=\sli{\omega_Q}{\omega_R}=\sls{\omega_Q}{\omega_R}&=&\srm{\omega_Q}{\omega_R}\,.
\end{eqnarray*}
\end{thm}
\begin{proof}
The theorem follows immediately from Proposition \ref{mre} and Theorem \ref{thm:error exponents}.
\end{proof}
 
\section{Concluding remarks}

We applied the results of \cite{HMO2} to the hypothesis testing problem of discriminating the local restrictions of two shift-invariant quasi-free states on a CAR algebra, and established the equality of various error exponents and the corresponding relative entropy-like quantities. In \cite{HMO2} the general problem was analyzed without any restriction on the relation of the supports of $\rho_n$ and $\sigma_n$, while here we assumed them to be equal, which removed some of the technical difficulties and allowed us to express our results in a more compact form. Furthermore, the expression \eqref{formula:psi} for $\psi$ provides explicitly (at least numerically) computable formulas for the 
mean relative entropy and mean Chernoff and Hoeffding bounds, that do not involve evaluations of limits. Though the relative entropy and the Chernoff bound are strictly positive in finite dimensions, this property does not necessarily hold for their asymptotic versions in general. In our case, however, formulas \eqref{formula:mre} and \eqref{formula:psi} show that the mean relative entropy or the mean Chernoff bound can only vanish if $\fv q=\fv r$ almost everywhere, i.e., if $\omega_Q=\omega_R$.

As is well-known in the literature of hypothesis testing (see e.g.~\cite{DZ,Han,NH}), the optimal error exponents are achieved by using the \ki{Neyman-Pearson tests} $S_{n,a},\,n\in\N$, where $a$ is some properly chosen real number and 
$S_{n,a}:=\{e^{-n^{\dimen}a}\D{\rho}_n-\D{\sigma}_n>0\}$, the spectral projection corresponding to the positive part of the spectrum of the self-adjoint operator $e^{-n^{\dimen}a}\D{\rho}_n-\D{\sigma}_n$. These tests are easily seen to be minimizers of $T_n\mapsto e^{-n^{\dimen}a}\alpha_n(T_n)+\beta_n(T_n),\,0\le T_n\le I_n$, as was pointed out e.g.~in \cite{Helstrom}. The results of \cite{HMO2} (Theorem 3.1, Corollary 4.5, Remark 4.6) combined with Proposition \ref{mre} of the present paper then give 
\begin{eqnarray*}
\lim_{n\to\infty}\frac{1}{n^{\dimen}}\log\bz e^{-n^{\dimen}a}\alpha_n(S_{n,a})+\beta_n(S_{n,a})\jz&=&-\vfi(a)\,,\ds\ds a\in\R,\\
\lim_{n\to\infty}\frac{1}{n^{\dimen}}\log\alpha_n(S_{n,a})&=&-\tilde\vfi(-a)\,,\ds\ds a>-\srm{\omega_R}{\omega_Q}\,,\\
\lim_{n\to\infty}\frac{1}{n^{\dimen}}\log\beta_n(S_{n,a})&=&-\vfi(a)\,,\ds\ds a<\srm{\omega_Q}{\omega_R}\,,
\end{eqnarray*}
where 
\begin{equation*}
\vfi(a):=\max_{0\le t\le 1}\{ta-\psi(t)\}\,,\ds\ds\ds
\tilde\vfi(a):=\max_{0\le t\le 1}\{ta-\psi(1-t)\}
\end{equation*}
are the polar functions of $\psi$ and $\tilde\psi(t):=\psi(1-t)$ on $[0,1]$. In \cite{HMO2} care had to be taken about some exceptional values of $a$ that may exist due to the possibility of $\psi$ being affine on certain intervals; however, this problem cannot occur for the quasi-free states we discussed in this paper. Indeed,
% a straightforward computation shows that $\psi''(t)=\frac{1}{(2\pi)^{\dimen}}\int_{[0,2\pi)^{\dimen}}\left[\E_{\vect{x},t}f_{\vect{x}}^2-\bz\E_{\vect{x},t}f_{\vect{x}}\jz^2\right]\dx,\, t\in\R$,
% % \begin{equation*}
% % \psi''(t)=\frac{1}{2\pi}\int_0^{2\pi}\left[\E_{x,t}f_{x}^2-\bz\E_{x,t}f_{x}\jz^2\right]\,\dx\,,\ds\ds\ds t\in\R\,,
% % \end{equation*}
% where 
% %\begin{equation*}
% $f_{\vect{x}}:\,\{0,1\}\to\R,\s f_{\vect{x}}(s):=\log(s+(-1)^s \fv q(\vect{x}))-\log(s+(-1)^s \fv r(\vect{x}))$ and $E_{\vect{x},t}$ is the expectation value with respect to the probability measure
% $p_{\vect{x},t}(0):=\fv q(\vect{x})^t\fv r(\vect{x})^{1-t}/S_{\vect{x},t},\s p_{\vect{x},t}(1):=(1-\fv q(\vect{x}))^t(1-\fv r(\vect{x}))^{1-t}/S_{\vect{x},t}$,
% % \begin{equation*}
% % p_{x,t}(0):=\frac{\fv q(x)^t\fv r(x)^{1-t}}{S_{x,t}},\ds\ds\ds\ds p_{x,t}(1):=\frac{(1-\fv q(x))^t(1-\fv r(x))^{1-t}}{S_{x,t}},
% % \end{equation*}
% with $S_{\vect{x},t}$ being the normalization factor.
%As a consequence, either 
one can easily see that in our case either $\psi$ is affine on $\R$ and $\fv q(\vect{x})=\fv r(\vect{x})$ for almost every $\vect{x}$ (and therefore the two states are equal), or $\psi''(t)>0$ for all $t\in\R$. The same result was obtained in \cite[Lemma 3.2]{HMO2} in an i.i.d.~setting and we conjecture this property (i.e., affinness on $\R$ vs. strict convexity) to hold in most cases of interest.

Stein's lemma generally treats the optimal exponential decay of the $\beta_n$'s under the constraint that the $\alpha_n$'s stay under some given constant bound. It seems, however, that in all cases when the exponential decay rate of these exponents can be identified, it coincides with the mean relative entropy \cite{BS,HP-1,ON}, which is the optimal exponent we obtained under the stronger constraint that the $\alpha_n$'s have to vanish asymptotically. Moreover, one can immediately verify from Theorem \ref{main result} that in the setting of the present paper 
\begin{eqnarray*}
\hlim{0}{\omega_Q}{\omega_R}=\hli{0}{\omega_Q}{\omega_R}=\hls{0}{\omega_Q}{\omega_R}&=& \hboundm{0}{\omega_Q}{\omega_R}\\
 &=&\srm{\omega_Q}{\omega_R},
 \end{eqnarray*}  
showing that one obtains the mean relative entropy as the optimal exponent even if the $\alpha_n$'s are required to vanish with an (arbitrarily slow) exponential speed. A similar result was obtained in a more general setting in \cite[Proposition 4.9]{HMO2}.

Finally, we remark that shift-invariant states on CAR($l^2(\Z)$) can be transferred in a one-to-one way to shift-invariant states on the spin chain $\C:=\otimes_{k=-\infty}^{\infty}\B(\iC^2)$ such that the local restrictions to CAR$\bz\spa\{\egy_{\{0\}},\ldots,\egy_{\{n-1\}}\}\jz$ are transferred to the local restrictions to $\otimes_{k=0}^{n-1}\B(\iC^2)$. This fact lies in the heart of determining the ground state of the XY-chain; see e.g.~\cite{LSM}, or \cite[Examples 5.2.21 and 6.2.14]{BR2} and \cite{FHM} for an overview.  As a consequence, the asymptotic hypothesis testing problem for the local restrictions of two shift-invariant quasi-free states can be interpreted as an asymptotic hypothesis testing problem for the local restrictions of the corresponding shift-invariant states on the spin chain $\C$.

\section*{Acknowledgments}

Partial funding was provided by PRESTO "Quanta and Information" in JST (T.O.),
Grant-in-Aid for Scientific Research (B)17340043 (F.H.), Grant-in-Aid for JSPS
Fellows 18\,$\cdot$\,06916, the Hungarian Research Grant
OTKA T068258 (M.M.), and the Belgian Interuniversity Attraction Poles Programme P6/02 (M.F.).

\end{document}